\batchmode
\documentclass[12pt,english]{article}
\usepackage{geometry}
\usepackage{color}
\usepackage{babel}
\usepackage{array}
\usepackage{units}
\usepackage{mathrsfs}
\usepackage{multirow}
\usepackage{amsmath}
\usepackage{amsthm}
\usepackage{amssymb}
\usepackage{setspace}
\usepackage[authoryear]{natbib}
\PassOptionsToPackage{normalem}{ulem}
\usepackage{ulem}
\usepackage[pdftex,unicode=true,pdfusetitle,
 bookmarks=true,bookmarksnumbered=false,bookmarksopen=false,
 breaklinks=false,pdfborder={0 0 0},pdfborderstyle={},backref=false,colorlinks=true]
 {hyperref}
\hypersetup{
 citecolor=blue, linkcolor=blue, urlcolor=blue}

\makeatletter

\usepackage{etoolbox}

\global\long\def\ExpOfD#1#2{\mathbb{E}_{#2}\left[#1\right]}%
\global\long\def\ExpOf#1{\ExpOfD{#1}{}}%

\global\long\def\PrOfD#1#2{\Pr_{#2}\left[#1\right]}%
\global\long\def\PrOf#1{\PrOfD{#1}{}}%

\global\long\def\MyMathop#1{\mathop{\mathrm{#1}}}%
\global\long\def\argmax{\MyMathop{argmax}}%

\global\long\def\xPrime#1{#1^{\prime}}%

\global\long\def\CEof#1#2{\mathrm{\textsc{CE}}_{#1}\left(#2\right)}%

\global\long\def\GrowthRateOf#1#2{g_{#1}\left(#2\right)}%
\global\long\def\GrowthRateOf#1#2{\mathcal{R}{}_{#1}\left(#2\right)}%
\global\long\def\GrowthRateOf#1#2{gr_{#1}\left(#2\right)}%

\usepackage{dsfont}

\global\long\def\PrefDist{\Phi}%
\global\long\def\PrefDistOpt{\PrefDist^{\star}}%

\global\long\def\CRRAr#1{\textsc{CRRA}_{#1}}%

\global\long\def\prefOver{\succcurlyeq}%
\global\long\def\strictPrefOver{\succ}%

\global\long\def\Indifferent{\sim}%

\providecommand{\tabularnewline}{\\}

\theoremstyle{remark}
\newtheorem{rem}{\protect\remarkname}
\theoremstyle{plain}
\newtheorem{fact}{\protect\factname}
\theoremstyle{plain}
\newtheorem{thm}{\protect\theoremname}
\theoremstyle{plain}
\newtheorem{cor}{\protect\corollaryname}

\@ifundefined{date}{}{\date{}}
\usepackage[dvipsnames]{xcolor}
\usepackage{colortbl}

\usepackage{tikz}
\usepackage{pgfplots}\pgfplotsset{compat = newest}

\usepackage{subcaption}

\usepackage{enumitem}\setlist[itemize,enumerate]{nosep}

\makeatother

\providecommand{\corollaryname}{Corollary}
\providecommand{\factname}{Fact}
\providecommand{\remarkname}{Remark}
\providecommand{\theoremname}{Theorem}

\begin{document}
\global\long\def\GMof#1{GM\left(#1\right)}%
\global\long\def\HMof#1{HM\left(#1\right)}%
\global\long\def\GMof#1{\mathbb{GM}\left[#1\right]}%
\global\long\def\HMof#1{\mathbb{HM}\left[#1\right]}%

\title{Evolutionary Foundation for Heterogeneity in Risk Aversion{\renewcommand{\thefootnote}{\relax}\renewcommand{\footnotemark}{\relax}\thanks{We thank the editor (Tilman B\"orgers), an anonymous referee, and
participants at the Neuroeconomics and the Biological Basis of Economics
conference at Simon-Fraser University, the Bar-Ilan Game and Economic
theory seminar, the 16\protect\textsuperscript{\uline{th}} Meeting
of the Society for Social Choice and Welfare, and the LEG2022 conference
at the IMT School for Advanced Studies (Lucca) for various helpful
comments. We gratefully acknowledge the financial support of the European
Research Council (\#677057), US--Israel Bi-national Science Foundation
(\#2020022), and Israel Science Foundation (\#2566/20, \#2443/19,
and \#1626/18).}}}
\date{Jan 24\textsuperscript{{\large{}\uline{th}}}, 2023\\
{\small{}Final pre-print of a manuscript accepted for publication
in }\emph{\small{}Journal of Economic Theory}{\small{}.}}
\author{Yuval Heller\thanks{Department of Economics, Bar-Ilan University. Email: \protect\href{mailto:yuval.heller@biu.ac.il}{yuval.heller@biu.ac.il}.}\and
Ilan Nehama\thanks{Department of Exact Sciences, Haifa University. Email: \protect\href{mailto:inehama@sci.haifa.ac.il}{inehama@sci.haifa.ac.il}
or \protect\href{mailto:ilan.nehama@mail.huji.ac.il}{ilan.nehama@mail.huji.ac.il}.}}
\maketitle
\begin{abstract}
We examine the evolutionary basis for risk aversion with respect to
aggregate risk. We study populations in which agents face choices
between alternatives with different levels of aggregate risk. We show
that the choices that maximize the long-run growth rate are induced
by a heterogeneous population in which the least and most risk-averse
agents are indifferent between facing an aggregate risk and obtaining
its linear and harmonic mean for sure, respectively. Moreover, approximately
optimal behavior can be induced by a simple distribution according
to which all agents have constant relative risk aversion, and the
coefficient of relative risk aversion is uniformly distributed between
zero and two.


\end{abstract}

\section{Introduction}

Our understanding of risk attitudes can be sharpened by considering
their evolutionary basis in situations in which agents face choices
that affect their number of offspring (see \citealp{robson2011evolutionary},
for a survey). \citet{lewontin1969population} shows that idiosyncratic
risk (independent across individuals) induces a higher expected long-run
growth rate (henceforth, abbreviated as \emph{growth rate}) than aggregate
risk (correlated across individuals) with the same marginal distribution
(i.e., for each individual the distribution of offspring induced is
the same). Specifically, the growth rate induced by an idiosyncratic-risk
lottery is equal to its arithmetic mean, while the growth rate induced
by an aggregate-risk lottery is equal to its geometric mean.

In his seminal paper, \citet{robson1996biological} shows that when
a population faces a choice between several aggregate-risk lotteries,
the optimal growth rate can be achieved by choosing the mixture of
these lotteries that maximizes the expected logarithm of the number
of offspring.\footnote{In this paper, we focus on characterizing the type that maximizes
the growth rate (and do not explicitly model the underlying dynamics).
This focus is motivated by the following sketched argument (which
is an adaptation of Fisher's Fundamental Theorem; see \citealp[Section 3.6]{weibull1997evolutionary}
for a textbook exposition). Consider a large population in which various
groups of agents have different heritable types, where the type determines
how the agent chooses between lotteries over the number of offspring
(importantly, the number of the agent's offspring is independent of
the types' frequencies in the population). Occasionally, new types
are introduced into the population following a genetic mutation. Observe
that the population share of agents of the type that induces the highest
growth rate grows until, in the long run, almost all agents are of
this type (see, e.g., \citealp{robson2011evolutionary}, for a detailed
argument of why natural selection induces agents to have types that
maximize the growth rate).}$^{,}$\footnote{See also \citet{robson2009evolution} and \citet{netzer2009evolution}
who study the evolution of risk attitudes and their impact on time
preferences, \citet{heller2014overconfidence} who argues that the
evolution of risk attitudes induces overconfidence, \citet{robatto2017biological}
who study choices that influence fertility rate in continuous time,
\citet{Robson-Samuelson-2019} who explore age-structured populations,
\citet{netzer2021endogenous} who argue that constrained optimal perception
affects people's risk attitudes and induces probability weighting,
\citet{robson2021evolved} who study the relation between aggregate
risk and the equity premium, and \citet{heller2021evolution} who
analyze heritable risk, which is correlated between an agent and her
offspring.} For example, consider a choice between an alternative that yields
2 offspring to each agent for sure (the \emph{safe} alternative) or
a \emph{risky} alternative yielding either 1 offspring or 5 offspring
to all agents who choose it, with equal probabilities. One can show
that the expected logarithm of the number of offspring is maximized
(and the optimal growth rate is achieved) by having $\nicefrac{1}{3}$
of the agents choosing the safe alternative, and the remaining $\nicefrac{2}{3}$
of the agents choosing the risky alternative. This mixture yields
an average number of offspring in the population that is either 4
or $\nicefrac{4}{3}$ in each generation. This heterogeneity in the
choices of the agents can be interpreted as \emph{bet-hedging }the
expected number of offspring in the population (\citealp{cohen1966optimizing,cooper1982adaptive,bergstrom1997storage}).
The existing literature typically does not specify how the optimal
bet-hedging mixtures are implemented.\footnote{\citet{mcnamara1995implicit} shows that the optimal bet-hedging can
be implemented when the agents maximize their relative fitness (namely,
the ratio between an agent\textquoteright s number of offspring and
the total number of offspring in her generation (see also~\citealp{grafen1999formal,curry2001decision,orr2007absolute}).
However, this maximization requires the agents to know the aggregate
behavior in the population, which was, arguably, implausible in most
of our evolutionary past.}

Note that each agent faces a non-convex choice between alternatives,
and thus she cannot hedge her own personal risk (for example, there
is no choice that yields her either 4 or $\nicefrac{4}{3}$ offspring
in the above example). One way to implement the optimal level of bet-hedging
is to induce each agent to randomly choose her alternative according
to the probability that maximizes the population's growth rate. Observe
that this behavior is complicated in the sense that it is not induced
by a von Neumann--Morgenstern (vNM) utility. For example, in the
above scenario she strictly prefers choosing a lottery to either of
the two alternatives. Hence, when facing a choice between two alternatives,
she needs to evaluate her utility also from various lotteries over
them.

In this paper, we introduce a new mechanism for implementing the optimal
growth rate. We show that the optimal growth rate is induced by a
heterogeneous population of utility maximizers in which agents have
different levels of risk aversion with respect to the aggregate risk
(and all agents are risk neutral with respect to idiosyncratic risk).
In this population, the most risk-averse agent is indifferent between
obtaining a risky lottery $\boldsymbol{y}$ and obtaining the harmonic
mean of~$\boldsymbol{y}$ for sure, while the least risk-averse agent
is risk neutral, that is, indifferent between obtaining the risky
lottery~$\boldsymbol{y}$ and obtaining the arithmetic mean of~$\boldsymbol{y}$
for sure. Moreover, we show that a nearly-optimal growth rate can
be achieved by a simple distribution of vNM utilities, according to
which all agents have constant relative risk aversion, and the risk
coefficient is uniformly distributed between zero and two. 

\paragraph{Highlights of the Model}

We consider a continuum population with asexual reproduction. Each
agent lives for a single generation, during which she faces a choice
between two lotteries over the number of offspring: a \emph{safe }(degenerate)\emph{
alternative} that gives $\mu$ offspring for sure, and a \emph{risky
alternative} $\boldsymbol{y}$ with aggregate risk.\footnote{A restriction of the present analysis is that it assumes that in any
generation there exists a single risky alternative (as discussed in
Section \ref{sec:Discussion}). We leave for future research the important
question of how to extend our analysis to a more general setup with
multiple (non-degenerate) risky alternatives.} Nature induces a distribution of risk preferences with regard to
aggregate risk, according to which each agent in the population (deterministically)
chooses between the risky alternative and the safe alternative.

\paragraph{Main Results}

If nature were limited to endowing all agents with the same preference,
then it would be optimal for all agents to evaluate any risky alternative
$\boldsymbol{y}$ as having a certainty equivalent of its geometric
mean. However, heterogeneous populations can induce a substantially
higher growth rate, because heterogeneity in risk aversion allows
the population to hedge the aggregate risk by enabling scenarios in
which only a portion of the population (the less risk-averse agents)
chooses the risky alternative.

Theorem~\ref{thm:AlphaStar Characterization} characterizes the optimal
share of agents that choose the risky alternative, that is, the share
that maximizes the long-run growth rate. In particular, it shows that
all agents should choose the safe option iff $\ExpOf{\boldsymbol{y}}\leqslant\mu$,
and all agents should choose the risky option iff the harmonic mean
of $\boldsymbol{y}$, $\HMof{\boldsymbol{y}}$, satisfies $\HMof{\boldsymbol{y}}\geqslant\mu$.
Moreover, we characterize the optimal preference distribution (Theorem~\ref{thm:PhiStar-CE})
as follows: (1) the least risk-averse agent in the population is risk
neutral, (2) the most risk-averse agent in the population has harmonic
utility, i.e., she evaluates any risky alternative $\boldsymbol{y}$
as having a certainty equivalent of its harmonic mean, and (3) all
agents in the population should be risk averse, but less risk averse
than the harmonic utility.

The optimal distribution of preferences is quite complicated. By contrast,
our numerical analysis shows that a nearly-optimal growth rate can
be achieved by a simple distribution of vNM utilities with constant
relative risk aversion ($\CRRAr{}$), where the relative risk coefficient
is uniformly distributed between zero (risk neutrality) and two (harmonic
utility). The predictions of our model fit reasonably well with the
empirical works on the distribution of risk attitudes in the population,
as discussed in Section~\ref{sec:Discussion}.

\paragraph{Structure}

The paper is structured as follows. In Section~\ref{sec:Model} we
describe the model. Section~\ref{sec:Results} presents the analytic
results, which are supplemented by a numerical analysis in Section~\ref{sec:Numeric-Analysis}.
We conclude with a discussion in Section~\ref{sec:Discussion}.

\section{Model\label{sec:Model}}

Consider a continuum population with an initial mass one. Reproduction
is asexual. Time is discrete, indexed by $t\in\mathbb{N}$. Each agent
lives a single time period (which is interpreted as a generation).
In each time period, each agent in the population faces a choice between
two alternatives, where each alternative is a lottery over the number
of offspring. The first alternative (henceforth, the \emph{safe alternative})
yields all agents who choose it with the same number of offspring,
$\mu$. The second alternative bears aggregate risk (henceforth, the
\emph{risky alternative}). That is, all agents who choose the risky
alternative have the same number of offspring, but this number is
a random variable $\boldsymbol{y}$ with finite support $supp\left(\boldsymbol{y}\right)=\left\{ y_{1},...,y_{n}\right\} \in\mathbb{R}_{+}$
and distribution $\PrOf{\boldsymbol{y}=y_{i}}=p_{i}$.

Let $\mathscr{Y}$ denote the set of risky alternatives (i.e., distributions
over nonnegative numbers). A risky alternative $\boldsymbol{y}\in\mathscr{Y}$
is \emph{degenerate} if $\left|supp\left(\boldsymbol{y}\right)\right|=1$,
and we identify it with the respective constant that it yields.
\begin{rem}
Our model can capture a more general setup that also includes idiosyncratic
risk (which is independent between different agents). It is well known
(see, e.g., \citealp{robson1996biological}) that the long-run growth
rate is the same for all idiosyncratic random variables with the same
expectation.\footnote{This is implied by an exact law of large numbers for continuum populations.
We refer the interested reader to~\citet{duffie2012exact} (and the
citations therein) for details on how the exact law of large numbers
is formalized in a related setup.} In this extended setup, one should interpret the safe lottery as
one that gives an expected number of offspring $\mu$, where the number
of offspring of each agent choosing the safe alternative might be
the result of an arbitrary idiosyncratic lottery with expectation~$\mu$.
Similarly, the risky alternative should be interpreted as a random
variable~$\boldsymbol{y}$ over the expected number of offspring
(i.e., for each of the realizations $\boldsymbol{\boldsymbol{y}}=y_{i}$
of the risky alternative, the number of offspring of each agent choosing
the risky alternative might be any idiosyncratic random variable with
expectation $y_{i}$). An additional assumption in this extended model
is that all agents are endowed with risk neutrality with respect to
idiosyncratic risk (i.e., they care only about the expected number
of offspring).
\end{rem}

\paragraph{Growth rate}

Let $\boldsymbol{w}\left(t\right)$ denote the size of the population
in time $t$. Let $\GrowthRateOf{\alpha}{\boldsymbol{y},\mu}$ denote
the long-run growth rate (henceforth, \emph{growth rate}) of a population
in which in each generation a share $\alpha$ of the population chooses
the risky alternative $\boldsymbol{y}$ and the remaining agents obtain
the safe option $\mu$. It is well known (see, e.g.,~\citealp{robson1996biological})
that the growth rate $\GrowthRateOf{\alpha}{\boldsymbol{y},\mu}$
equals the geometric mean of $\alpha\boldsymbol{y}+\left(1-\alpha\right)\mu$:
\begin{equation}
\begin{array}{rl}
\GrowthRateOf{\alpha}{\boldsymbol{y},\mu} & \equiv\lim_{T\rightarrow\infty}\sqrt[T]{\frac{\boldsymbol{w}\left(T\right)}{\boldsymbol{w}\left(1\right)}}=\GMof{\alpha\boldsymbol{y}+\left(1-\alpha\right)\mu}=\\
 & =\prod_{i\leqslant n}\left(\alpha y_{i}+\left(1-\alpha\right)\mu\right)^{p_{i}}=e^{\sum_{i\leqslant n}ln\left(p_{i}\cdot\left(\alpha y_{i}+\left(1-\alpha\right)\mu\right)\right)}\text{.}
\end{array}\label{eq:growth-rate}
\end{equation}

The intuition for Equation~(\ref{eq:growth-rate}) is as follows.
Let $\boldsymbol{z}\left(t\right)=\frac{\boldsymbol{w}\left(t+1\right)}{\boldsymbol{w}\left(t\right)}$
be the mean number of offspring in generation~$t$; i.e., $\boldsymbol{z}\left(t\right)$
is a sequence of i.i.d. variables that are distributed like the random
variable $\alpha\boldsymbol{y}+\left(1-\alpha\right)\mu$. Hence,
the size of the population at time $T$ equals
\[
\begin{array}{l}
\frac{\boldsymbol{w}\left(T\right)}{\boldsymbol{w}\left(1\right)}=\prod_{t<T}\frac{\boldsymbol{w}\left(t+1\right)}{\boldsymbol{w}\left(t\right)}=e^{\left(\sum_{t\leqslant T}\ln\left(\boldsymbol{z}\left(t\right)\right)\right)}\\
\quad\Rightarrow\lim\limits _{T\rightarrow\infty}\sqrt[T]{\frac{\boldsymbol{w}\left(T\right)}{\boldsymbol{w}\left(1\right)}}=\lim\limits _{T\rightarrow\infty}e^{\left(\frac{1}{T}\sum\limits _{t<T}\ln\left(\boldsymbol{z}\left(t\right)\right)\right)}\stackrel{{\scriptstyle \left(\star\right)}}{=}e^{\ExpOf{\ln\left(\alpha\boldsymbol{y}+\left(1-\alpha\right)\mu\right)}}=\prod\limits _{i\leqslant n}\left(\alpha y_{i}+\left(1-\alpha\right)\mu\right)^{p_{i}}\text{,}
\end{array}
\]
where the equality marked by $\left(\star\right)$ is implied by the
law of large numbers.

Let $\alpha^{\star}\in\left[0,1\right]$ be the share of agents who
choose the risky alternative that maximizes the long-run growth rate:
\begin{equation}
\alpha^{\star}\left(\boldsymbol{y},\mu\right)=\argmax\limits _{\alpha\in\left[0,1\right]}\left(\GrowthRateOf{\alpha}{\boldsymbol{y},\mu}\right)\text{.}\label{eq:optimal-alpha}
\end{equation}
We show in Theorem~\ref{thm:AlphaStar Characterization} that $\alpha^{\star}\left(\boldsymbol{y},\mu\right)$
is unique.

\paragraph*{Preferences}

Each agent is endowed with a preference over the lotteries, i.e.,
a linear order $\prefOver$ over the set $\mathscr{Y}$ (and we use
the notation $\Indifferent$ for indifference). That is, Agent~$a$
chooses the risky alternative iff $\boldsymbol{y}\strictPrefOver_{a}\mu$
(the tie-breaking rule that is applied when $\boldsymbol{y}\Indifferent_{a}\mu$
has no impact on our results since it holds for a share of measure
zero of the population). A preference $\prefOver$ is \emph{regular}
if it satisfies the following two mild assumptions: (1) monotonicity
of the safe alternatives: $\mu<\xPrime{\mu}$ implies that $\mu\prec\xPrime{\mu}$,
and (2) any risky alternative has a certainty equivalent: for any
$y\in\mathscr{Y}$, there exists a safe alternative $\mu$ such that
$\boldsymbol{y}\Indifferent\mu$.

Let $\mathscr{U}$ denote the set of regular preferences. Observe
that any regular preference $\prefOver$ can be represented by a \emph{certainty
equivalent function }$CE_{\prefOver}\colon\mathscr{Y}\rightarrow\mathbb{R}_{+}$,
which evaluates each risky alternative in terms of the equivalent
safe alternative (i.e., $CE_{\prefOver}\left(y\right)=\mu$ iff $\boldsymbol{y}\Indifferent\mu$).

We assume that nature endows the population with a distribution $\PrefDist$
of regular preferences, and that each agent uses her preference to
choose an alternative. A distribution of regular preferences $\PrefDist$
induces a choice function $\alpha_{\PrefDist}\colon\mathscr{Y}\times\mathbb{R}_{+}\rightarrow\left[0,1\right]$,
which describes the share of agents who choose the risky option for
any pair of alternatives.

A distribution of regular preferences $\PrefDistOpt$ is \emph{optimal}
if for any $\boldsymbol{y}\in\mathscr{Y}$ and $\mu\in\mathbb{R}_{+}$
it maximizes the growth rate, i.e.,
\begin{equation}
\alpha_{\PrefDistOpt}\left(\boldsymbol{y},\mu\right)=\alpha^{\star}\left(\boldsymbol{y},\mu\right)=\argmax\limits _{\alpha\in\left[0,1\right]}\left(\GrowthRateOf{\alpha}{\boldsymbol{y},\mu}\right)\text{.}\label{eq:optimal-alpha-Pref}
\end{equation}

\paragraph{Dynamic interpretation}

Our interpretation of the model is that in each generation, all agents
face a choice between the same pair of a safe alternative and a risky
alternative. The choice changes over time, such that in different
generations there might be different alternatives to choose from.
The optimal distribution of preferences is required to maximize the
expected long-run growth rate, and it is not hard to see that this
is equivalent to maximizing the growth rate $\GrowthRateOf{\alpha}{\boldsymbol{y},\mu}$
for any combination of a risky alternative~$\boldsymbol{y}$ and
a safe alternative~$\mu$ (which appears with non-zero frequency).

Further, observe that the optimal distribution is required to maximize
the growth rate among all possible choice profiles (including choice
functions that are not consistent with having preferences). Thus,
our results show that the restriction of agents to regular preferences
over the set of risky alternatives does not decrease the maximal feasible
growth rate (see~\citealp{robson2001biological,robson2011evolutionary}
for arguments in favor of endowing agents with utilities).

\paragraph{Utility}

A common way to represent the preference of Agent~$a$ is to use
a \emph{utility function}, $U_{a}\colon\mathscr{Y}\rightarrow\mathbb{R}_{+}$,
such that Agent~$a$ strictly prefers an alternative~$\boldsymbol{y}\in\mathscr{Y}$
to another alternative $\boldsymbol{\xPrime y}\in\mathscr{Y}$ iff
$U_{a}\left(\boldsymbol{y}\right)>U_{a}\left(\boldsymbol{\xPrime y}\right)$.
We note that for a given regular preference $\prefOver$, its certainty
equivalent function $CE_{\prefOver}$ is in particular a utility function
that represents~$\prefOver$.

A preference $\prefOver$ is a \emph{vNM (von Neumann--Morgenstern)
preference} if it has an expected utility representation, that is,
if there exists a Bernoulli utility function $u\colon\mathbb{R}_{+}\rightarrow\mathbb{R}$
such that $\prefOver$ is represented by the utility function $\ExpOf{u\left(\boldsymbol{y}\right)}=\sum_{i}p_{i}\cdot u\left(y_{i}\right)$
for any $\boldsymbol{y}\in\mathscr{Y}$.

\paragraph{Risk aversion and constant relative risk aversion ($\protect\CRRAr{}$)
preferences}

A preference $\prefOver$ is \emph{risk averse} (resp., risk neutral)
if $CE_{\prefOver}\left(\boldsymbol{y}\right)<\ExpOf{\boldsymbol{y}}$
(resp., $CE_{\prefOver}\left(\boldsymbol{y}\right)=\ExpOf{\boldsymbol{y}}$)
for any nondegenerate risky alternative~$\boldsymbol{y}\in\mathscr{Y}$.
A preference $\prefOver$ is \emph{more risk averse} than a preference
$\xPrime{\prefOver}$ if $\CEof{\prefOver}{\boldsymbol{y}}<\CEof{\xPrime{\prefOver}}{\boldsymbol{y}}$
for any nondegenerate risky alternative~$\boldsymbol{y}\in\mathscr{Y}$.

For any $\rho\geqslant0$, let $\CRRAr{\rho}$ denote the constant
relative risk aversion preference with relative risk coefficient $\rho$,
i.e., the expected utility preference defined by
\begin{equation}
\widehat{u}_{\rho}\left(y_{i}\right)=\begin{cases}
\frac{y_{i}^{1-\rho}-1}{1-\rho} & \rho\neq1\\
\ln\left(y_{i}\right) & \rho=1
\end{cases}\text{.}\label{eq:CRRA}
\end{equation}

Let $\HMof{\boldsymbol{y}}$ and $\GMof{\boldsymbol{y}}$ denote the
harmonic and geometric means of $\boldsymbol{y}$, respectively, 
\[
\HMof{\boldsymbol{y}}=\left(\ExpOf{\boldsymbol{y}^{-1}}\right)^{-1}=\frac{1}{\sum_{i}\nicefrac{p_{i}}{y_{i}}},\,\,\,\,\GMof{\boldsymbol{y}}=\prod_{i}y_{i}^{p_{i}}\text{.}
\]

It is well known that:
\begin{enumerate}
\item $\HMof{\boldsymbol{y}}\leqslant\GMof{\boldsymbol{y}}\leqslant\ExpOf{\boldsymbol{y}}$
with strict inequality whenever $\boldsymbol{y}$ is nondegenerate.
\item Under the utilities $\CRRAr 0$, $\CRRAr 1$, and $\CRRAr 2$, the
certainty equivalent values of any risky alternative $\boldsymbol{y}\in\mathscr{Y}$
are its arithmetic, geometric, and harmonic means, respectively. Hence,
we also refer to $\CRRAr 0$ as the \emph{linear utility}, to $\CRRAr 1$
as the \emph{logarithmic utility}, and to $\CRRAr 2$ as the \emph{harmonic
utility}.
\end{enumerate}
Last, a distribution $\PrefDist$ of regular preferences is \emph{monotone}
if its support is a chain with respect to the strong risk aversion
order; i.e., for any two preferences $\prefOver$ and $\xPrime{\prefOver}$
in the support of $\PrefDist$, $\prefOver$ is either strictly more
risk averse or strictly less risk averse than $\xPrime{\prefOver}$.

\section{Results\label{sec:Results}}

Observe that if nature is limited to a homogeneous population in which
all agents have the same risk preference, then the maximal long-run
growth rate is attained by the logarithmic utility $\CRRAr 1$, according
to which the certainty equivalent of a risky alternative is its geometric
mean (see, e.g.,~\citealp{lewontin1969population}). This is an immediate
corollary of Equation~\eqref{eq:growth-rate} and the definition
of $\CRRAr{}$ preferences.
\begin{fact}
\label{fact:logUtil is optimal}For any $\boldsymbol{y}\in\mathscr{Y}$
and $\mu\in\mathbb{R}_{+}$, $\GrowthRateOf 1{\boldsymbol{y},\mu}\geqslant\GrowthRateOf 0{\boldsymbol{y},\mu}\Leftrightarrow\widehat{U}_{1}\left(\boldsymbol{y}\right)\geqslant\widehat{U}_{1}\left(\mu\right)$.
\end{fact}
Our analysis is motivated by the fact that heterogeneous populations
in which agents differ in the extent of their risk aversion can induce
substantially higher growth rate, because the heterogeneity allows
the population to hedge the aggregate risk by having only the more
risk-averse agents choose the risky alternative. For example, if agents
face a choice between a safe alternative $\mu$ yielding one offspring
to each agent or a risky alternative $\boldsymbol{y}$ yielding either
4 offspring or 0.25 offspring with equal probabilities, then any homogeneous
population in which all agents share the same risk preference (with
a deterministic tie-breaking rule) yields a growth rate of 1 (because
both $\GrowthRateOf 0{\boldsymbol{y},\mu}=\mu=1$ and $\GrowthRateOf 1{\boldsymbol{y},\mu}=\GMof{\boldsymbol{y}}=4^{0.5}\cdot0.25^{0.5}=1$).
By contrast, a heterogeneous population in which agents differ in
the extent of their risk aversion such that half the population choose
the risky alternative $\boldsymbol{y}$ and the others choose the
safe alternative $\mu$ induces a substantially higher growth rate
of
\[
\GrowthRateOf{0.5}{\boldsymbol{y},\mu}=\GMof{0.5\cdot\boldsymbol{y}+0.5\cdot1}=2.5^{0.5}\cdot0.625^{0.5}=1.25\text{.}
\]

Our first result characterizes the optimal share of agents that choose
the risky alternative.\footnote{Similar results to Theorem~\ref{thm:AlphaStar Characterization}
have been presented in related setups (see, e.g., the relative fitness
condition in \citealp{Brennan2012}). For completeness we present
a short proof.}
\begin{thm}
\label{thm:AlphaStar Characterization}Fix $\boldsymbol{y}\in\mathscr{Y}$
and $\mu\in\mathbb{R}_{+}$. Then:
\begin{enumerate}
\item $\alpha^{\star}\left(\boldsymbol{y},\mu\right)=\argmax_{\alpha\in\left[0,1\right]}\left(\GrowthRateOf{\alpha}{\boldsymbol{y},\mu}\right)$
is unique.
\item $\alpha^{\star}\left(\boldsymbol{y},\mu\right)=0$ iff $\ExpOf{\boldsymbol{y}}\leqslant\mu$,
and

$\alpha^{\star}\left(\boldsymbol{y},\mu\right)=1$ iff $\HMof{\boldsymbol{y}}\geqslant\mu$.
\item If $\mu\in\left(\HMof{\boldsymbol{y}},\ExpOf{\boldsymbol{y}}\right)$,
then $\alpha^{\star}\left(\boldsymbol{y},\mu\right)\in\left(0,1\right)$
is the unique solution to the following equation:
\begin{equation}
\HMof{1+x\cdot\left(\nicefrac{\boldsymbol{y}}{\mu}-1\right)}=1\text{.}\label{eq:AlphaStarIsAsolutionToThis}
\end{equation}
\end{enumerate}
\end{thm}
\begin{proof}
The long-run growth rate when a share $\alpha$ of the population
chooses $\boldsymbol{y}$ is (see Equation~\eqref{eq:growth-rate})
\[
\GrowthRateOf{\alpha}{\boldsymbol{y},\mu}=\GMof{\alpha\cdot\boldsymbol{y}+\left(1-\alpha\right)\cdot\mu}=e^{\ExpOf{\ln\left(\alpha\cdot\boldsymbol{y}+\left(1-\alpha\right)\cdot\mu\right)}}\text{.}
\]

Hence, $\GrowthRateOf{\alpha}{\boldsymbol{y},\mu}$ is maximized iff
$\ln\left(\GrowthRateOf{\alpha}{\boldsymbol{y},\mu}\right)=\ExpOf{\ln\left(\alpha\cdot\boldsymbol{y}+\left(1-\alpha\right)\cdot\mu\right)}$
is maximized, and since
\[
\begin{array}[t]{rl}
\frac{d}{d\alpha}\ln\left(\GrowthRateOf{\alpha}{\boldsymbol{y},\mu}\right) & =\ExpOf{\frac{\boldsymbol{y}-\mu}{\alpha\cdot\boldsymbol{y}+\left(1-\alpha\right)\cdot\mu}}\\
\frac{d^{2}}{d^{2}\alpha}\ln\left(\GrowthRateOf{\alpha}{\boldsymbol{y},\mu}\right) & =-\ExpOf{\left(\frac{\boldsymbol{y}-\mu}{\alpha\cdot\boldsymbol{y}+\left(1-\alpha\right)\cdot\mu}\right)^{2}}<0\text{,}
\end{array}
\]
there is exactly one maximizer for $\GrowthRateOf{\alpha}{\boldsymbol{y},\mu}$
in $\left[0,1\right]$, and the following three statements hold:
\begin{itemize}
\item $\alpha^{\star}\left(\boldsymbol{y},\mu\right)=0$ iff $\frac{d}{d\alpha}\ln\left(\GrowthRateOf{\alpha}{\boldsymbol{y},\mu}\right)|_{\alpha=0}=\ExpOf{\nicefrac{\boldsymbol{y}}{\mu}}-1\leqslant0$,
i.e., $\ExpOf{\boldsymbol{y}}\leqslant\mu$.
\item $\alpha^{\star}\left(\boldsymbol{y},\mu\right)=1$ iff $\frac{d}{d\alpha}\ln\left(\GrowthRateOf{\alpha}{\boldsymbol{y},\mu}\right)|_{\alpha=1}=1-\ExpOf{\nicefrac{\mu}{\boldsymbol{y}}}\geqslant0$,
i.e., $\mu\leqslant\HMof{\boldsymbol{y}}$.
\item If $\mu\in\left(\HMof{\boldsymbol{y}},\ExpOf{\boldsymbol{y}}\right)$,
then $\alpha^{\star}\left(\boldsymbol{y},\mu\right)\in\left(0,1\right)$
is the unique solution to
\[
\ExpOf{\frac{\boldsymbol{y}-\mu}{x\cdot\boldsymbol{y}+\left(1-x\right)\cdot\mu}}=0\text{.}
\]
Noting that for $x\neq0$, 
\[
\begin{array}[t]{rl}
\ExpOf{\frac{\boldsymbol{y}-\mu}{\mu+x\cdot\left(\boldsymbol{y}-\mu\right)}} & =\frac{1}{x}\cdot\left(1-\ExpOf{\frac{\mu}{\mu+x\cdot\left(\boldsymbol{y}-\mu\right)}}\right)=\frac{1}{x}\cdot\left(1-\left(\HMof{1+x\cdot\left(\nicefrac{\boldsymbol{y}}{\mu}-1\right)}\right)^{-1}\right)\end{array}\text{,}
\]
we get that $\alpha^{\star}\left(\boldsymbol{y},\mu\right)\in\left(0,1\right)$
is the unique solution to
\[
\HMof{1+x\cdot\left(\nicefrac{\boldsymbol{y}}{\mu}-1\right)}=1\text{\text{.}\qedhere}
\]
\end{itemize}
\end{proof}
Given a distribution $\PrefDist$ of regular preferences and a risky
alternative $\boldsymbol{y}\in\mathscr{Y}$, we define $\PrefDist_{\boldsymbol{y}}$
to be the distribution of certainty equivalent values of $\boldsymbol{y}$
in the population.Our next result characterizes the optimal distribution
of risk preferences. Specifically, it shows that for any risky alternative
$\boldsymbol{y}$, $\boldsymbol{\left(1\right)}$ the support of $\PrefDist_{\boldsymbol{y}}$
is the range between $\boldsymbol{y}$'s harmonic mean and $\boldsymbol{y}$'s
arithmetic mean, and $\boldsymbol{\left(2\right)}$ the $\lambda$-median
of $\PrefDist_{\boldsymbol{y}}$ is the unique solution to a simple
equation.
\begin{thm}
\label{thm:PhiStar-CE}

Let $\PrefDist$ be a distribution of regular preferences. Then, $\PrefDist$
is optimal iff for any risky alternative $\boldsymbol{y}\in\mathscr{Y}$,
the cumulative density function (CDF) of $\PrefDist_{\boldsymbol{y}}$
is
\[
\textrm{CDF}_{\PrefDist_{\boldsymbol{y}}}\left(x\right)=1-\alpha^{\star}\left(\boldsymbol{y},x\right)\text{.}
\]

In particular, for any (nondegenerate) risky alternative~$\boldsymbol{y}\in\mathscr{Y}$,
\begin{enumerate}
\item The support of $\PrefDist_{\boldsymbol{y}}$ is $\left[\HMof{\boldsymbol{y}},\ExpOf{\boldsymbol{y}}\right]$.
\item For any $\lambda\in\left(0,1\right)$ the $\lambda$-median of $\PrefDist_{\boldsymbol{y}}$
is the unique solution to
\[
\HMof{\lambda+\left(1-\lambda\right)\cdot\nicefrac{\boldsymbol{y}}{x}}=1\text{.}
\]
\end{enumerate}
\end{thm}
Note that by Theorem~\ref{thm:AlphaStar Characterization}, $\left(1-\alpha^{\star}\left(\boldsymbol{y},x\right)\right)$
is indeed a CDF for any (nondegenerate) risky alternative $\boldsymbol{y}$.
$\alpha^{\star}\left(\boldsymbol{y},x\right)$ equals one when $x\leqslant\HMof{\boldsymbol{y}}$,
equals zero when $x\geqslant\ExpOf{\boldsymbol{y}}$, and equals the
solution to $\HMof{1+\alpha\cdot\left(\nicefrac{\boldsymbol{y}}{x}-1\right)}=1$
otherwise. The function $x\mapsto\alpha^{\star}\left(\boldsymbol{y},x\right)$
is continuous and strictly downward monotone in $x\in\left(\HMof{\boldsymbol{y}},\ExpOf{\boldsymbol{y}}\right)$
since the function $\left(x,\alpha\right)\mapsto\HMof{1+\alpha\cdot\left(\nicefrac{\boldsymbol{y}}{x}-1\right)}$
is continuous, with a bounded domain, and strictly downward monotone
in $x$. That is, $1-\alpha^{\star}\left(\boldsymbol{y},x\right)$
equals zero for $x\leqslant\HMof{\boldsymbol{y}}$, equals one for
$x\geqslant\ExpOf{\boldsymbol{y}}$, and is continuous and strictly
upward monotone in between.
\begin{proof}
Let $\PrefDist$ be a distribution of regular preferences. An agent
prefers $\boldsymbol{y}$ to a safe alternative $\mu$ iff her certainty
equivalent value of $\boldsymbol{y}$ is higher than $\mu$, which
holds for a share $1-\textrm{CDF}_{\PrefDist_{\boldsymbol{y}}}\left(\mu\right)$
of the population, i.e.,
\[
\alpha_{\PrefDist}\left(\boldsymbol{y},\mu\right)=1-\textrm{CDF}_{\PrefDist_{\boldsymbol{y}}}\left(\mu\right)\text{.}
\]
Hence, $\PrefDist$ is optimal iff for any $\boldsymbol{y}\in\mathscr{Y}$
and $\mu\in\mathbb{R}_{+}$, 
\[
\alpha_{\PrefDist}\left(\boldsymbol{y},\mu\right)=\alpha^{\star}\left(\boldsymbol{y},\mu\right)\text{, i.e., }\textrm{CDF}_{\PrefDist_{\boldsymbol{y}}}\left(\mu\right)=1-\alpha^{\star}\left(\boldsymbol{y},\mu\right)\text{.}
\]
In particular, by Theorem~\ref{thm:AlphaStar Characterization},
for any risky alternative $y\in\mathscr{Y}$, the support of $\PrefDist_{\boldsymbol{y}}$
is $\left[\HMof{\boldsymbol{y}},\ExpOf{\boldsymbol{y}}\right]$. Moreover,
for any $\lambda\in\left(0,1\right)$, the $\lambda$-median of $\PrefDist_{\boldsymbol{y}}$,
$m_{\lambda}$, satisfies
\[
\lambda=\textrm{CDF}_{\PrefDist_{\boldsymbol{y}}}\left(m_{\lambda}\right)=1-\alpha^{\star}\left(\boldsymbol{y},m_{\lambda}\right)\text{.}
\]
Therefore,
\[
\HMof{1+\left(1-\lambda\right)\cdot\left(\frac{\boldsymbol{y}}{m_{\lambda}}-1\right)}=1
\]
and $m_{\lambda}$ is a solution to
\[
\HMof{\lambda+\left(1-\lambda\right)\cdot\nicefrac{\boldsymbol{y}}{x}}=1\text{.}
\]

Lastly, $\HMof{\lambda+\left(1-\lambda\right)\cdot\nicefrac{\boldsymbol{y}}{x}}$
is strictly downward monotone in $x$, and hence there is a unique
solution to $\HMof{\lambda+\left(1-\lambda\right)\cdot\nicefrac{\boldsymbol{y}}{x}}=1\text{.}$
\end{proof}
\begin{cor}
\label{cor:unique-montone} By Theorem~\ref{thm:PhiStar-CE}, the
following is the unique monotone optimal distribution of regular preferences
$\PrefDistOpt$. We index the agents by $\left[0,1\right]$ and define
the preference of Agent~$a\in\left[0,1\right]$ by defining her certainty
equivalent value for any risky alternative $\boldsymbol{y}\in\mathscr{Y}$
to be:
\begin{itemize}
\item $\HMof{\boldsymbol{y}}$ if $a=0$,
\item $\ExpOf{\boldsymbol{y}}$ if $a=1$, and
\item the unique solution to
\[
\HMof{a+\left(1-a\right)\cdot\nicefrac{\boldsymbol{y}}{x}}=1\quad\text{ otherwise.}
\]
\end{itemize}
\end{cor}
\begin{rem}
We note that $\CEof a{\boldsymbol{y}}$ is continuous in the parameter~$a$.
It is easy to see that $\HMof{\boldsymbol{y}}$ is the limit solution
when $a\rightarrow0$ for 
\[
\HMof{a+\left(1-a\right)\cdot\nicefrac{\boldsymbol{y}}{x}}=1\text{.}
\]
 For $a\rightarrow1^{-}$, noticing that
\[
\begin{array}{l}
\left(\HMof{a+\left(1-a\right)\cdot\nicefrac{\boldsymbol{y}}{x}}\right)^{-1}=\\
\qquad=1-\left(1-a\right)\cdot\left(\frac{1}{x}\cdot\ExpOf{\frac{\boldsymbol{y}}{1+\left(1-a\right)\cdot\left(\frac{\boldsymbol{y}}{x}-1\right)}}-\ExpOf{\frac{1}{1+\left(1-a\right)\cdot\left(\frac{\boldsymbol{y}}{x}-1\right)}}\right)\text{,}
\end{array}
\]
we get that $\CEof a{\boldsymbol{y}}$ satisfies
\[
\CEof a{\boldsymbol{y}}=\ExpOf{\frac{\boldsymbol{y}}{1+\left(1-a\right)\cdot\left(\frac{\boldsymbol{y}}{\CEof a{\boldsymbol{y}}}-1\right)}}\cdot\left(\ExpOf{\frac{1}{1+\left(1-a\right)\cdot\left(\frac{\boldsymbol{y}}{\CEof a{\boldsymbol{y}}}-1\right)}}\right)^{-1}\text{,}
\]
and hence 
\[
\lim_{a\rightarrow1^{-}}\CEof a{\boldsymbol{y}}=\ExpOf{\boldsymbol{y}}\text{.}
\]
\end{rem}
We note that the behavior of the least and the most risk-averse agents
in $\PrefDistOpt$ is simple and intuitive. The choices of the least
risk-averse agent, Agent~1, are consistent with $\CRRAr 0$ (risk
neutrality), and the choices of the most risk-averse agent, Agent~0,
are consistent with $\CRRAr 2$. By contrast, for any $a\in\left(0,1\right)$
and $\boldsymbol{y}\in\mathscr{Y}$, the choices of Agent~$a$ are
consistent with $\CRRAr{\rho_{a}\left(\boldsymbol{y}\right)}$ for
some $\rho_{a}\left(\boldsymbol{y}\right)\in\left(0,2\right)$; the
dependency of $\rho_{a}\left(\boldsymbol{y}\right)$ on $\boldsymbol{y}$
makes the representation of the preferences of these agents more cumbersome
and in particular, as we show in Appendix~\ref{appendix:The optimal PrefDist can not be represented using vNM utility},
they do not have an expected utility representation.

\section{Numerical Analysis\label{sec:Numeric-Analysis}}

Section \ref{subsec:setup-simulation} presents a Monte Carlo simulation
that we use to evaluate what percentage of the theoretically optimal
growth rate is induced by various distributions of utilities. The
numerical results (Section~\ref{subsec:numeric-results}) show that
simple distributions of preferences, in which all agents have $\CRRAr{}$
utilities and the relative risk coefficient is distributed between
zero (risk neutrality) and two (harmonic utility), achieve (on average)
99.86\% of the optimal long-run growth rate. In Section \ref{subsec:numeric-interpretation},
we interpret these result as suggesting that simple distributions
of $\CRRAr{}$ utilities could be the result of the evolutionary process.

\subsection{Setup and Simulation\label{subsec:setup-simulation}}

The code of the simulation is detailed in the \href{https://www.dropbox.com/s/bhzqlj2mohtkntt/Heller\%20\%26\%20Nehama\%20-\%20Evolutionary\%20Foundation\%20for\%20Heterogeneity\%20in\%20Risk\%20Aversion\%20-\%20Supplementary\%20Material\%20-\%20code.zip}{online supplementary material}.

\paragraph{Distributions of utilities}

We compare 15 distributions of utilities:
\begin{enumerate}
\item Five homogeneous populations in which all agents have the same utility:
\begin{enumerate}
\item \emph{Extreme risk loving}: All agents always choose the risky alternative
(as long as $\PrOf{\boldsymbol{y}>\mu}\neq0$).
\item \emph{Extreme risk aversion}: All agents always choose the safe alternative
(as long as $\PrOf{\boldsymbol{y}<\mu}\neq0$).
\item \emph{Risk neutrality ($\CRRAr 0$)}: All agents evaluate risky choices
by their arithmetic mean.
\item \emph{Logarithmic utility ($\CRRAr 1$)}: All agents evaluate risky
choices by their geometric mean.
\item \emph{Harmonic utility ($\CRRAr 2$)}: All agents evaluate risky choices
by their harmonic mean.
\end{enumerate}
\item Two classes of heterogeneous populations with monotone distributions.
In each class, each agent is endowed with a value $\beta\in\left[0,1\right]$
(each class includes 5 distributions of $\beta$ as detailed below).
All classes have the property characterized by Corollary~\ref{cor:unique-montone},
namely, that the most and the least risk-averse agents (corresponding
to $\beta=1$ and $\beta=0$, respectively) evaluate a risky alternative
$\boldsymbol{y}$ as having a certainty equivalent value of the harmonic
mean and arithmetic mean of $\boldsymbol{y}$, respectively.

The behavior of the agent endowed with value $\beta\in\left[0,1\right]$
in each class is as follows:
\begin{enumerate}
\item \emph{Heterogeneous constant relative risk aversion: populations}:
All agents have $\CRRAr{2\beta}$ preferences where $\beta$'s distribution
is detailed below.
\item \emph{Heterogeneous weighted-average populations}: All agents evaluate
risky alternatives as a weighted average of their harmonic mean and
arithmetic means: $\CEof{\beta}{\boldsymbol{y}}=\beta\cdot\HMof{\boldsymbol{y}}+\left(1-\beta\right)\cdot\ExpOf{\boldsymbol{y}}$,
where $\beta$'s distribution is detailed as follows.\footnote{For any $\beta\neq0,1$, this preference cannot be represented using
expected utility. Consider the lottery $\boldsymbol{L}=\begin{cases}
6 & \nicefrac{1}{2}\\
2 & \nicefrac{1}{2}
\end{cases}$, and the two degenerate lotteries $\boldsymbol{M}=4-\beta$ and $\boldsymbol{N}=4$.
Then, %
\begin{minipage}[t]{0.7\columnwidth}%
\begin{itemize}
\item $\CEof{\beta}{\boldsymbol{L}}=\CEof{\beta}{\boldsymbol{M}}=4-\beta.$
\item $\CEof{\beta}{\nicefrac{1}{2}\boldsymbol{L}+\nicefrac{1}{2}\boldsymbol{N}}=4-\frac{4\beta}{7}$
\item $\CEof{\beta}{\nicefrac{1}{2}\boldsymbol{M}+\nicefrac{1}{2}\boldsymbol{N}}=\frac{1}{2\cdot\left(8-\beta\right)}\cdot\left(64-16\beta+\beta^{2}-\beta^{3}\right)$
\end{itemize}
\end{minipage}

\noindent{}and hence Agent~$\beta$ is indifferent between $\boldsymbol{L}$
and $\boldsymbol{M}$ but not between $\nicefrac{1}{2}\boldsymbol{L}+\nicefrac{1}{2}\boldsymbol{N}$
and $\nicefrac{1}{2}\boldsymbol{M}+\nicefrac{1}{2}\boldsymbol{N}$
(in violation of the \emph{independence axiom} of vNM) and, in particular,
the preference of Agent~$\beta$ cannot be represented by expected
utility.}
\end{enumerate}
We use five beta distributions for $\beta\in\left[0,1\right]$ for
the two classes (as demonstrated in Figure~\ref{Fig:Betas}):
\begin{enumerate}
\item Uniform distribution: $\beta\sim Beta\left(1,1\right)$.
\item Unimodal distribution: $\beta\sim Beta\left(2,2\right)$.
\item Bimodal distribution: $\beta\sim Beta\left(0.5,0.5\right)$.
\item Positively skewed distribution: $\beta\sim Beta\left(2,4\right)$.
\item Negatively skewed distribution: $\beta\sim Beta\left(4,2\right)$.
\end{enumerate}
\end{enumerate}
\newcommand{\BetaCDFint}[3]{
\begin{tikzpicture}%
	\begin{axis}[
		xmin = 0, xmax = 1,
		ymin = 0, 
		xtick distance = .2,
		ytick distance = #3,
		minor x tick num = 1,
		minor y tick num = 0,
		width = \textwidth, height = 0.5\textwidth,
	]
\addplot[
	domain = 0:1,
	samples = 200,
	smooth,
	red,thin,
] {x^(#1-1)*(1-x)^(#2-1) *   (#1+#2-1)! / ( (#1-1)! * (#2-1)! )};
\end{axis}\end{tikzpicture}}
\newcommand{\BetaCDFintAndHalf}[3]{
\begin{tikzpicture}%
	\begin{axis}[
		xmin = 0, xmax = 1,
		ymin = 0, 
		xtick distance = .2,
		ytick distance = #3,
		minor x tick num = 1,
		minor y tick num = 1,
		width = \textwidth, height = 0.5\textwidth,
	]
\addplot[
	domain = 0.001:.999,
	samples = 200,
	smooth,
	red,thin,
] {x^(#1-1) * (1-x)^(#2-1) * 0.31830988618 * 4^(#1+#2-1) *
		(#1+#2-1)! * (#1-.5)! * (#2-.5)! / 
		((2*#1 -1)! * (2*#2 -1)!)
};
\end{axis}\end{tikzpicture}}

\begin{figure}[t]
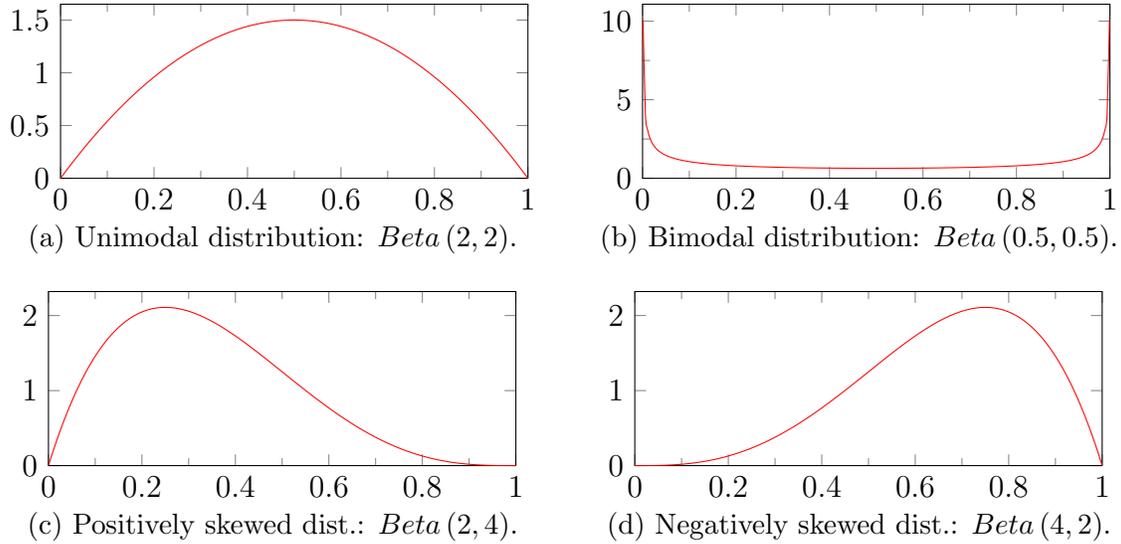
\centering{}

\begin{subfigure}[t]{0.5\linewidth}\centering{}\BetaCDFint{2}{2}{.5}\vspace{-2mm}\caption{\label{Fig:Beta-2-2}Unimodal
distribution: $Beta\left(2,2\right)$.}\end{subfigure}\begin{subfigure}[t]{0.5\linewidth}\centering{}\BetaCDFintAndHalf{.5}{.5}{5}\vspace{-2mm}\caption{\label{Fig:Beta-.5-.5}Bimodal
distribution: $Beta\left(0.5,0.5\right)$.}\end{subfigure}

\bigskip{}

\begin{subfigure}[t]{0.5\linewidth}\centering{}\BetaCDFint{2}{4}{1}\vspace{-2mm}\caption{\label{Fig:Beta-2-4}Positively
skewed dist.: $Beta\left(2,4\right)$.}\end{subfigure}\begin{subfigure}[t]{0.5\linewidth}\centering{}\BetaCDFint{4}{2}{1}\vspace{-2mm}\caption{\label{Fig:Beta-4-2}Negatively
skewed dist.: $Beta\left(4,2\right)$.}\end{subfigure}

\caption{\label{Fig:Betas}Probability Density Function (PDF) of
$Beta\left(\alpha,\beta\right)$.}

\end{figure}

\paragraph{Description of the simulation}

\global\long\def\yLow{\ell}%
\global\long\def\yHigh{h}%
The simulation evaluates the performance of each distribution of utilities
over $10.7M$ choices between a safe alternative $\mu$ and a binary
risky alternative $\boldsymbol{y}$ yielding either a low realization
$\yLow$ or a high realization $\yHigh$. We run the following scenario
for the distribution of the risky alternative and the safe alternative
in each generation (the alternatives in different generations are
independent of each other):
\begin{itemize}
\item In each generation, the two alternatives are defined by three independent
uniform random numbers $p,q,r\in\left[0,1\right]$, where:\footnote{Equivalently, $p,\yLow,$ and $h$ are independent given $\mu$ and
are sampled as follows: $p\in_{\mathbb{U}}\left[0,1\right]$, $\yLow\in_{\mathbb{U}}\left[0,\mu\right]$,
and $\yHigh\in\left[\mu,\infty\right)$ with the inverse-uniform distribution
with parameters $\left\langle 0,1\right\rangle $ ($F\left(x\right)=1-\nicefrac{\mu}{x}\,\,;\,\,f\left(x\right)=\nicefrac{\mu}{x^{2}}$).}
\begin{itemize}
\item $p$ is the probability of the risky alternative yielding its high
realization $\yHigh$: $p=\PrOf{\boldsymbol{y}=\yHigh}$.
\item $q$ is the ratio between the low realization of the risky alternative
and the value of the safe alternative: $q=\nicefrac{\yLow}{\mu}$.
\item $r$ is the ratio between the value of the safe alternative and the
high realization of the risky alternative: $r=\nicefrac{\mu}{\yHigh}$.
\end{itemize}
\end{itemize}
Without loss of generality, we normalize the value of the safe alternative
to be $\mu=1$. In each simulation run we calculate the theoretically
optimal growth rate $\GrowthRateOf{\alpha^{\star}}{\boldsymbol{y},\mu}$,
and then evaluate the percentage of this optimal growth rate achieved
by each of the 15 distributions of utilities. Finally, we calculate
the geometric mean of this percentage for each distribution over all
the simulation runs, which evaluates the relative performance of each
distribution (in terms of its long-run growth rate) in a setup in
which the risky and safe alternatives change from one generation to
the next.

\subsection{Numerical Results\label{subsec:numeric-results}}

The results are summarized in Table~\ref{cor:unique-montone}. The
optimal growth rate in our setup is 1.428 (which is calculated as
the geometric mean of the growth rate achieved in each generation).
We evaluate the performance of each distribution of preferences according
to the decline in the relative growth rate, i.e., according to the
percentage of the optimal growth rate that is lost under this distribution
of preferences. The best homogeneous population is the one in which
all agents have logarithmic utility, and it achieves a loss of $2.1
$ relative to the optimal growth rate. Heterogeneous $\CRRAr{}$ populations
reduce this loss substantially to less than $1
$ (which is better than what can be achieved by the heterogeneous weighted-average
populations). Moreover, heterogeneous $\CRRAr{}$ populations in which
$\beta$ is distributed uniformly (or according to the unimodal distribution)
reduce this loss further to $0.15\%$.

\begin{table}[h]
\caption{Summary of results of simulation runs ($10.7M$ generations).\label{table:Summary of Results of Simulation Runs}}

\global\long\def\Colored#1#2{#1#2}%
\renewcommand{\Colored}[2]{
	\ifnumgreater{#1}{19}{%
		\cellcolor{OrangeRed!80}#1.#2\%
	}{\ifnumgreater{#1}{5}{%
		\cellcolor{Dandelion}#1.#2\%
	}{\ifnumgreater{#1}{0}{%
		\cellcolor{Yellow!80}#1.#2\%
	}{\ifnumgreater{#2}{90}{%
		\cellcolor{SpringGreen}#1.#2\%
	}{
		\cellcolor{LimeGreen}#1.#2\%
}}}}}%

\hspace{-5mm}%
\begin{tabular}{|c|c|c|c|c|}
\hline 
\multirow{2}{*}{Class} &
\multirow{2}{*}{Distribution of $\beta$} &
\begin{tabular}{@{}c@{}}
Empirical\tabularnewline
mean of $\alpha$\tabularnewline
\end{tabular} &
\begin{tabular}{@{}c@{}}
Long-run\tabularnewline
growth rate\tabularnewline
\end{tabular} &
\begin{tabular}{@{}c@{}}
Relative\tabularnewline
growth rate\tabularnewline
loss\tabularnewline
\end{tabular}\tabularnewline
\cline{3-5} \cline{4-5} \cline{5-5} 
 &  & $\ExpOf{\alpha}$ &
$\GMof{\GrowthRateOf{\alpha}{\boldsymbol{y},\mu}}$ &
$1-\frac{\GMof{\GrowthRateOf{\alpha}{\boldsymbol{y},\mu}}}{\GMof{\GrowthRateOf{\alpha^{\star}}{\boldsymbol{y},\mu}}}$\tabularnewline
\hline 
\hline 
Optimal &
(Corollary \ref{cor:unique-montone}) &
0.500 &
1.4251 &
$\Colored 0{00}$\tabularnewline
\hline 
\hline 
\multirow{5}{*}{%
\begin{tabular}{@{}c@{}}
Homogeneous\tabularnewline
populations\tabularnewline
\end{tabular}} &
Extreme risk loving &
1.0 &
0.9949 &
$\Colored{30}2$\tabularnewline
\cline{2-5} \cline{3-5} \cline{4-5} \cline{5-5} 
 & Extreme risk aversion &
0.0 &
1.0000 &
$\Colored{29}8$\tabularnewline
\cline{2-5} \cline{3-5} \cline{4-5} \cline{5-5} 
 & Risk neutrality &
0.644 &
1.3152 &
$\Colored 77$\tabularnewline
\cline{2-5} \cline{3-5} \cline{4-5} \cline{5-5} 
 & %
\begin{tabular}{@{}c@{}}
Logarithmic utility\tabularnewline
($\CRRAr 1$)\tabularnewline
\end{tabular} &
0.499 &
1.3949 &
$\Colored 21$\tabularnewline
\cline{2-5} \cline{3-5} \cline{4-5} \cline{5-5} 
 & %
\begin{tabular}{@{}c@{}}
Harmonic utility\tabularnewline
($\CRRAr 2$)\tabularnewline
\end{tabular} &
0.357 &
1.3172 &
$\Colored 76$\tabularnewline
\hline 
\hline 
\multirow{5}{*}{%
\begin{tabular}{c}
Heterogeneous\tabularnewline
$\CRRAr{}$\tabularnewline
populations\tabularnewline
\end{tabular}} &
Uniform &
0.500 &
1.4232 &
$\Colored 0{14}$\tabularnewline
\cline{2-5} \cline{3-5} \cline{4-5} \cline{5-5} 
 & Unimodal &
0.500 &
1.4231 &
$\Colored 0{15}$\tabularnewline
\cline{2-5} \cline{3-5} \cline{4-5} \cline{5-5} 
 & Bimodal &
0.500 &
1.4211 &
$\Colored 0{29}$\tabularnewline
\cline{2-5} \cline{3-5} \cline{4-5} \cline{5-5} 
 & Positively skewed &
0.551 &
1.4118 &
$\Colored 0{93}$\tabularnewline
\cline{2-5} \cline{3-5} \cline{4-5} \cline{5-5} 
 & Negatively skewed &
0.448 &
1.4117 &
$\Colored 0{94}$\tabularnewline
\hline 
\hline 
\multirow{5}{*}{%
\begin{tabular}{@{}c@{}}
Heterogeneous\tabularnewline
weighted-average\tabularnewline
populations\tabularnewline
\end{tabular}} &
Uniform &
0.530 &
1.4054 &
$\Colored 14$\tabularnewline
\cline{2-5} \cline{3-5} \cline{4-5} \cline{5-5} 
 & Unimodal &
0.535 &
1.3983 &
$\Colored 19$\tabularnewline
\cline{2-5} \cline{3-5} \cline{4-5} \cline{5-5} 
 & Bimodal &
0.524 &
1.4103 &
$\Colored 10$\tabularnewline
\cline{2-5} \cline{3-5} \cline{4-5} \cline{5-5} 
 & Positively skewed &
0.577 &
1.3757 &
$\Colored 35$\tabularnewline
\cline{2-5} \cline{3-5} \cline{4-5} \cline{5-5} 
 & Negatively skewed &
0.491 &
1.4054 &
$\Colored 14$\tabularnewline
\hline 
\end{tabular}
\end{table}

\subparagraph{Robustness check}

To check the robustness of our results, we tested other parameter
distributions ($30$ additional distributions in total) as follows:
\begin{itemize}
\item By taking the probability, $\PrOf{\boldsymbol{y}=\yHigh}$, and the
two ratios, $\nicefrac{\GMof{\boldsymbol{y}}}{\mu}$ and $\nicefrac{\mu}{\ExpOf{\boldsymbol{y}}}$,
to be three i.i.d. uniformly distributed random numbers in $\left[0,1\right]$. 
\item By conditioning the two distributions on the event $\left[\GMof{\boldsymbol{y}}\leqslant\mu\leqslant\ExpOf{\boldsymbol{y}}\right]$
and the events $\left[\frac{i-1}{k}\leqslant\frac{\mu-\GMof{\boldsymbol{y}}}{\ExpOf{\boldsymbol{y}}-\GMof{\boldsymbol{y}}}\leqslant\frac{i}{k}\right]$
for $k=2,\ldots,5$ and $i=1,\ldots,k$.
\end{itemize}
For all these distributions, we see similar qualitative results in
which the heterogeneous $\CRRAr{}$ populations outperform the other
populations and the optimal growth rate is approximated by a heterogeneous
$\CRRAr{}$ population under a simple distribution of the relative
risk parameter (uniform and unimodal).

\subsection{Interpretation of Results\label{subsec:numeric-interpretation}}

A shift from the optimal homogeneous population of agents (in which
everyone has a logarithmic utility, $\CRRAr 1$) to a heterogeneous
population (in which agents have different $\CRRAr{}$ utilities)
requires evolution to engineer a gene that manifest itself in different
levels of risk aversion for different people in the right proportion.
Our numerical result suggests that such a shift is reasonable as:
$\left(1\right)$ the added advantage of increasing the long-run growth
rate by 2\% per generation (when shifting from a homogeneous population
to a uniform (or unimodal) distribution of $\CRRAr{}$ utilities)
is substantial when compounded over many generations, and $\left(2\right)$
we show two different simple distributions of $\CRRAr{}$ utilities
(uniform and unimodal) that achieve similar nearly-optimal behavior,
which facilitates the genetic implementation, as the gene can implement
one of various nearly-optimal distributions.\footnote{Additional simulation runs suggest that any moderately-unimodal beta
distribution around 0.5 (i.e., any $\beta\sim Beta\left(x,x\right)$
for $1\leq x\leq2$) achieve a nearly-optimal growth rate.}

The numerical results further suggest that it would be difficult for
the evolutionary process to improve the growth rate relative to these
simple distributions of $\CRRAr{}$ utilities. A shift to the optimal
distribution of utilities (or to any other complicated mechanism that
implements exactly the optimal bet-hedging mixtures) would only improve
the mean growth rate by an additional 0.14\%. The cognitive capacity
of humans is limited by various constraints (see, e.g., \citealp{lieder2020resource}),
such as the energy consumption of the human brain (20\% of the body
energy) and the newborn baby's brain size, which is limited by the
birth canal. Allocating more cognitive resources to one activity has
a shadow cost induced by decreasing cognitive resources allocated
to other activities (a decrease that will reduce the agent's growth
rate). Although we do not have a formal model of these shadow costs,
it seems reasonable to conjecture that a shift from a simple $\CRRAr{}$
utility to a complicated non-vNM utility requires more cognitive resources,
and that their shadow costs outweigh the small benefit of 0.14\%.
This suggests that when taking into account cognitive costs, it might
be optimal for the evolutionary process to induce a simple distribution
of $\CRRAr{}$ utilities with relative risk aversion coefficients
between 0 and 2.

\section{Discussion\label{sec:Discussion}}

In what follows we discuss various aspects of our model and their
implications.

\paragraph{Empirical predictions}

Our model suggests that natural selection endowed the population with
(1) risk-averse preferences and (2) heterogeneity in the level of
risk aversion such that the agents' certainty equivalent values for
a given lottery are distributed between the lottery's harmonic mean
and its expectation, and that (3) the preference distribution can
be approximated by constant relative risk aversion utilities with
relative risk aversion between zero and two. Our model deals with
lotteries with respect to the number of offspring (fitness), but it
is plausible that people apply these endowed risk attitudes when dealing
with lotteries over money, which is what is typically tested in the
empirical literature.\footnote{A nonlinear relationship between consumption and fitness in our evolutionary
past might shift the optimal levels of risk aversion with respect
to money. Specifically, if the expected number of offspring is a concave
function of consumption, then the support of the optimal distribution
of relative risk aversion with respect to consumption will be shifted
to the right.} \citet{chiappori2011relative} rely on large panel data and show
that the elasticity of the relative risk aversion index with respect
to wealth is small and statistically insignificant, which supports
our first prediction of people having constant relative risk aversion
utilities. \citet{halek2001demography} relies on life insurance data
to estimate the distribution of the levels of relative risk aversion
in the population. Their data suggests that there is substantial heterogeneity
in the levels of relative risk aversion in the population, and that
about 80\% of the population have levels of relative risk aversion
between zero and two~(\citealp[Figure 1]{halek2001demography}).

\paragraph*{Multiple risky alternatives}

If there are multiple sources of risky alternatives, each with its
own shared risk (e.g., multiple foraging techniques, where agents
using the same foraging technique have correlated risk), then we implicitly
assume that agents use some decision rule to choose between the different
risky sources, and the single risky alternative in our model represents
a combination of these sources. For example, if there are several
independent and identically distributed risky alternatives $\boldsymbol{y}^{1},..,\boldsymbol{y}^{n}$,
then it is not hard to show that it is optimal for the population
to choose these alternatives with equal shares, which can be modeled
by the single risky alternative $\boldsymbol{y}=\frac{\boldsymbol{y}^{1}+...+\boldsymbol{y}^{n}}{n}$.
We do not analyze the general question of how to optimally diversify
risk among different sources of correlated risk.

\paragraph{Monomorphic heterogeneous populations}

The population in our model is a monomorphic population rather than
a polymorphic population; that is, all agents in our model have the
same genotype, which manifests itself in different degrees of risk
aversion in different people. In the biological literature, such phenomena
in which a single genotype induces heterogeneous behavior is known
as genetic expressivity (see, e.g., \citealp[Section 6.4]{Book-PenetranceAndExpressivity}),
and its usage in biological evolutionary models has been popularized
in \citet{grafen1990biological,grafen1990sexual}. Hence, the relative
performance of individuals does not affect the path of the evolutionary
dynamics.

\paragraph{Random expected utility}

Our interpretation of the optimal distribution of preferences in the
population is heterogeneity in the population; namely, some agents
are more risk averse than others. We note that the optimal distribution
can also be implemented by random expected utility (see, e.g.,~\citealp{gul2006random});
namely, each agent is endowed with the optimal distribution of preferences,
and in each decision problem each agent randomly applies one of these
preferences.

\section{Conclusion}

The existing literature has shown that when agents face a choice between
alternatives with various levels of aggregate risk, the optimal growth
rate induces the population to choose the mixture of these lotteries
that achieves the optimal level of bet-hedging. The main contribution
of this paper is to present a new mechanism that allows evolution
to implement the optimal level of bet-hedging. This is done by nature
inducing a heterogeneous population with different levels of risk
aversion, such that the most risk-averse agent is indifferent between
obtaining a risky lottery and obtaining its harmonic mean for sure,
while the least risk-averse agent is risk neutral. Although, the exactly-optimal
distribution of risk-averse utilities is quite complicated, we show
numerically that a nearly-optimal growth rate is induced by a population
of expected utility maximizers with constant relative risk aversion
preferences, in which the risk coefficient is distributed between
zero and two according to a simple distribution (uniform or unimodal).
Such a distribution might be optimal, if one takes into account the
cognitive costs.

\appendix

\section{Optimal Monotone Distribution Does Not Have an Expected Utility Representation\label{appendix:The optimal PrefDist can not be represented using vNM utility}}

Consider the following five lotteries:

\[
\begin{array}{l}
\bullet\,\boldsymbol{L}=\text{The degenerate (safe) lottery }3\qquad\bullet\,\boldsymbol{M}=\begin{cases}
1.5 & \nicefrac{3}{4}\\
20 & \nicefrac{1}{4}
\end{cases}\qquad\bullet\,\boldsymbol{N}=\begin{cases}
10 & \nicefrac{1}{2}\\
15 & \nicefrac{1}{2}
\end{cases}\\
\bullet\,\boldsymbol{X}=\nicefrac{1}{2}\boldsymbol{L}+\nicefrac{1}{2}\boldsymbol{N}=\begin{cases}
3 & \nicefrac{1}{2}\\
10 & \nicefrac{1}{4}\\
15 & \nicefrac{1}{4}
\end{cases}\hfill\bullet\,\boldsymbol{Y}=\nicefrac{1}{2}\boldsymbol{M}+\nicefrac{1}{2}\boldsymbol{N}=\begin{cases}
1.5 & \nicefrac{3}{8}\\
10 & \nicefrac{1}{4}\\
15 & \nicefrac{1}{4}\\
20 & \nicefrac{1}{8}
\end{cases}\hfill
\end{array}
\]

Here, the median agent prefers~$\boldsymbol{L}$ to~$\boldsymbol{M}$
and prefers $\boldsymbol{Y}=\nicefrac{1}{2}\boldsymbol{M}+\nicefrac{1}{2}\boldsymbol{N}$
to $\boldsymbol{X=}\nicefrac{1}{2}\boldsymbol{L}+\nicefrac{1}{2}\boldsymbol{N}$.
But this is a violation of the \emph{independence axiom} of vNM, and
in particular, the preference of the median agent cannot be represented
by expected utility.\footnote{By Theorem~\ref{thm:PhiStar-CE}, the certainty equivalent value
of the median agent for a lottery $\boldsymbol{y}\in\mathscr{Y}$
is the unique solution to $\HMof{1+\nicefrac{\boldsymbol{y}}{x}}=2$.}
\begin{itemize}
\item The median agent's certainty equivalent value for $\boldsymbol{M}$
is $\approx2.54$, and hence she prefers~$\boldsymbol{L}$ to~$\boldsymbol{M}$.
\item Her certainty equivalent value for $\boldsymbol{X}$ is $\approx6.04$,
and hence she prefers the safe option~$6.1$ to~$\boldsymbol{X}$.
\item Her certainty equivalent value for $\boldsymbol{Y}$ is $\approx6.19$,
and hence she prefers~$\boldsymbol{Y}$ to the safe option~$6.1$,
and~$\boldsymbol{Y}$ to~$\boldsymbol{X}$.
\end{itemize}

\end{document}